\documentclass[12pt,a4paper]{article}
\usepackage{amssymb}
\usepackage{amsfonts}
\usepackage[onehalfspacing]{setspace}

\newtheorem{definition}{Definition}

\newtheorem{proposition}{Proposition}

\newenvironment{proof}[1][Proof]{\noindent\textbf{#1.} }{\ \rule{0.5em}{0.5em}}
\input{tcilatex}
\begin{document}

\title{A Simple Model of Monetary Policy under Phillips-Curve Causal
Disagreements\thanks{%
Financial support by ERC Advanced Investigator grant no. 692995 is
gratefully acknowledged. I thank In-Koo Cho, Stephane Dupraz and Morten Ravn
for helpful comments.}}
\author{Ran Spiegler\thanks{%
Tel Aviv University, University College London and CFM. URL:
http://www.tau.ac.il/\symbol{126}rani. E-mail: rani@tauex.tau.ac.il.}}
\maketitle

\begin{abstract}
I study a static textbook model of monetary policy and relax the
conventional assumption that the private sector has rational expectations.
Instead, the private sector forms inflation forecasts according to a
misspecified subjective model that disagrees with the central bank's (true)
model over the causal underpinnings of the Phillips Curve. Following the
AI/Statistics literature on Bayesian Networks, I represent the private
sector's model by a direct acyclic graph (DAG). I show that when the private
sector's model reverses the direction of causality between inflation and
output, the central bank's optimal policy can exhibit an attenuation effect
that is sensitive to the noisiness of the true inflation-output
equations.\bigskip \bigskip \pagebreak
\end{abstract}

\section{Introduction}

Monetary policy is a classic example of an economic interaction in which
policy makers' ability to achieve their objectives depends on the accuracy
of agents' forecasts of certain economic variables. Consider a textbook
static model of the type originated by Kydland and Prescott (1977) and Barro
and Gordon (1983). In such a model, the central bank controls a policy
variable that affects inflation. The private sector forms an inflation
forecast, possibly after observing some signal regarding the central bank's
information or its decision. Private-sector expectations are relevant
because they affect the joint realization of inflation and output, via some
kind of an \textquotedblleft expectations-augmented\textquotedblright\
Phillips curve. It follows that monetary policy involves \textquotedblleft 
\textit{expectations management}\textquotedblright . To quote Woodford
(2003, p. 15):\medskip 

\begin{quote}
\textquotedblleft \ldots successful monetary policy is not so much a matter
of effective control of overnight interest rates as it is of shaping market
expectations of the way in which interest rates, inflation and income are
likely to evolve\ldots \textquotedblright \medskip
\end{quote}

Conventional models constrain the central bank's ability to manage
expectations by assuming that the private sector has \textquotedblleft
rational expectations\textquotedblright\ - i.e., it fully understands the
statistical regularities in its environment, and thus forms an unbiased
inflation forecast conditional on its information. To put it differently,
the private sector shares the central bank's (true) model of the
macroeconomy and uses it to form its beliefs. In Evans and Honkapohja
(2005), Thomas Sargent refers to this feature as \textquotedblleft communism
of models\textquotedblright . In this paper, I revisit the textbook model
and relax \textquotedblleft model communism\textquotedblright , by assuming
that the private sector forms its forecast according to a different model
than the central bank's (true) model. The key difference between the two
models is in their \textit{causal treatment of the Phillips Curve}.

The modern history of macroeconomic thought has included a number of forms
of Phillips Curves, which differ in their dependent and R.H.S variables.
Popular descriptions of the Phillips Curve often impose a causal
interpretation on these forms. For instance, when the dependent variable is
inflation, the causal story is about how changes in unemployment affect
inflation. When inflation expectations enter the R.H.S, they are presented
as another cause of inflation. When the dependent variable is unemployment
or output, the story is about how inflation (possibly in addition to
inflationary expectations) affects real variables. Cogley and Sargent
(2005), following King and Watson (1994), describe this as a distinction
between two econometric identification procedures, which they term
\textquotedblleft Keynesian direction of fit\textquotedblright\ and
\textquotedblleft classical direction of fit\textquotedblright . Taking a
broader perspective, Hoover (2001) interprets debates in the history of
monetary theory (and other areas of macroeconomics) in terms of
disagreements over causality.

Therefore, if we wish to model how the private sector forms\ inflation
forecasts using some kind of a Phillips Curve, it seems reasonable to assume
that it will involve a causal interpretation. Even if the true model is not
fundamentally causal (in the sense that the variables are determined in
equilibrium by some system of simultaneous equations with influences in both
directions), popular conceptions of the Phillips Curve and how it is
integrated with a broader macroeconomic model are likely to be causal in
nature.

Following Spiegler (2016,2020a), I adopt the formalism of Bayesian networks
(Cowell et al. (1999), Pearl (2009), Koller and Friedman (2009)) and
represent the private sector's causal model with a directed acyclic graph
(DAG). Nodes in the graph represent economic variables and directed links
represent postulated causal influences. The DAG can be interpreted in at
least two ways. First, it may capture intuitive narratives about causal
relations between the relevant variables. Second, it may represent an
explicit formal model of the kind that professional forecasters sometimes
employ, consisting of a recursive system of equations. The private sector
fits its subjective model to the steady-state distribution, and then relies
on this \textquotedblleft estimated model\textquotedblright\ to produce an
inflation forecast. When the model is misspecified, the inflation forecast
systematically distorts the true expected inflation conditional on the
private sector's information.

I incorporate this model of private-sector expectations into a variant on
the familiar static textbook model of central-bank policy. The central bank
aims to minimize the expectation of a quadratic loss function of the
deviations of output and inflation from their exogenous targets. The output
target is fixed, whereas the inflation target is random. The central bank
chooses a policy instrument in response to its inflation target, trading off
the deviation of inflation from its moving target against the deviation of
output from its fixed target. A Phillips Curve links these two deviations.
Output deviations depend on how the private sector's inflation forecast
responds to its information. The key question is therefore how the private
sector's narrative regarding the Phillips Curve affects the responsiveness
of the private sector's inflation forecast, and in turn the responsiveness
of the central bank to the inflation target.

Two qualitative insights emerge from the analysis. First, the private
sector's wrong causal model can result in \textit{rigid} inflation
forecasts, which in turn leads to monetary-policy attenuation - that is, the
central bank's equilibrium response to the inflation target is weaker than
in the rational-expectations benchmark. Second, the \textit{variance of the
noise terms} in the equations for inflation and output (one of which is a
Phillips Curve) can play an important role in the rigid-forecast effect,
when the private sector's model reverses the direction of causality between
these two variables.

When macroeconomists talk about \textquotedblleft
flatness\textquotedblright\ of the Phillips relation, there is ambiguity
about whether this means a weak yet precise relation or a noisy relation. In
a finite sample, a noisy Phillips relation will appear as a
\textquotedblleft cloud\textquotedblright\ of data points, such that the
Phillips effect may end up being statistically insignificant - hence,
effectively regarded in much the same way as a weak, precise relation.
However, when the private sector has rational expectations, the two have
very different effects: the variance of the noise terms is irrelevant for
the private sector's forecast, hence also for the central bank's optimal
policy. This is no longer the case when the private sector's forecasts are
based on a misspecified model. Using examples, I show that the variance of
the noise terms in the inflation-output equations has a subtle effect on the
private sector's inflation forecast and consequently on central-bank
policy.\medskip 

\noindent \textit{Related literature}

\noindent The Bayesian-network approach to modeling equilibrium non-rational
expectations was presented in Spiegler (2016). In Spiegler (2020a), I
presented a simple monetary-policy example to illustrate the question of
whether causal misperceptions can lead to systematically biased
expectations. In this paper, I adhere to the standard linear-normal
parameterization, such that the private sector's expectations are correct on
average. Therefore, systematically biased inflation forecasts are ruled out;
the departure from rational expectations lies in the under-sensitivity of
the private sector's inflation forecast to its information.

The Game Theory literature contains a number of related approaches to
equilibrium modeling with non-rational expectations, including analogy-based
expectations (Jehiel (2005)), \textquotedblleft cursed\textquotedblright\
beliefs (Eyster and Rabin (2005)) and Berk-Nash equilibrium (Esponda and
Pouzo (2016)). Spiegler (2020b) describes the relation between these
approaches and the Bayesian-network formalism.

Within the macroeconomic theory literature, there are a few precedents for
modeling monetary policy when the rational-expectations assumption is
relaxed. Evans and Honkapohja (2001) and Woodford (2013) review dynamic
macroeconomic models in which agents form non-rational expectations, and
explore implications for monetary policy. Orphanides and Williams (2007) and
Garcia-Schmidt and Woodford (2015) are examples of exercises in this
tradition. In particular, Orphanides and Williams (2007) study the
implications of learning-based private-sector expectations for the structure
of central-bank policy.

The most closely related equilibrium concept that is employed in the
macroeconomic literature is known as \textquotedblleft restricted
perceptions equilibrium\textquotedblright\ (Evans and Honkapohja (2001)),
which is based on a notion of coarse beliefs in the same spirit as Piccione
and Rubinstein (2003) and Jehiel (2005). Sargent (2001), Cho et al. (2002)
and Esponda and Pouzo (2016) study models in which it is the \textit{central
bank} that forms non-rational expectations, whereas the private sector is
modeled conventionally.

The idea that real-life monetary policy may be insufficiently responsive to
shocks has been discussed in the literature since at least Rudenbusch
(2001). A number of explanations have been proposed. One of them,
originating in Brainard (1967), is that policy attenuation is a response to
uncertainty about relevant parameters or to Knightian uncertainty regarding
the central bank's model (see Tetlow and Von zur Muehlen (2001) and Dupraz
et al. (2020) for critical discussions). This paper contributes to this
literature by proposing a new possible motive for monetary-policy
distortion: the private sector's misperception of the causal relation
between inflation and output.

\section{The Model}

The following is a variant on a familiar textbook static model of monetary
policy in the tradition of Kydland and Prescott (1977) and Barro and Gordon
(1983). It most closely resembles Sargent (1999) and Athey et al. (2005).

All variables in the model get real values. A central bank observes an
exogenous, normally distributed variable $\theta $ and chooses an action $a$%
. Subsequently, the private sector observes a Gaussian signal $t$ of $%
(\theta ,a)$ and forms an inflation forecast $e$. Following the realization
of $\theta ,a,t,e$, actual inflation $\pi $ and actual output $y$ are
realized according to the following pair of equations:%
\begin{eqnarray*}
\pi &=&a+\varepsilon \\
y &=&\pi -\lambda e+\eta
\end{eqnarray*}%
where $\lambda \in (0,1)$ is a constant that represents the extent to which
anticipated inflation offsets the real effect of actual inflation; and $%
\varepsilon $ and $\eta $ are statistically independent noise terms, $%
\varepsilon \sim N(0,\sigma _{\varepsilon }^{2})$ and $\eta \sim N(0,\sigma
_{\eta }^{2})$.

The central bank's objective is to minimize a standard quadratic loss
function:%
\[
L(\pi ,y,\theta )=y^{2}+(\pi -\theta )^{2} 
\]%
Thus, the exogenous state variable $\theta $ is interpreted as the central
bank's ideal inflation target. The central bank's strategy is a function $%
(a(\theta ))_{\theta }$ that assigns an action to every realization of $%
\theta $. The private sector's strategy is the forecast function $(e(t))_{t}$
that assigns an inflation forecast to every signal $t$.

Given the central bank's strategy and private sector's forecast function, we
can define a joint distribution $p$ over all six variables $\theta
,a,t,e,y,\pi $. The exogenous components of $p$ are the distribution over $%
\theta $, the conditional distribution over $t$ (given $\theta ,a$) and the
conditional distribution of $\pi $ and $y$ (given $\theta ,a,t,e$). The
endogenous components are the pair consisting of the central bank's strategy 
$(a(\theta ))_{\theta }$ and the private sector's forecast $(e(t))_{t}$.

We now arrive at the key definition of how the private sector forms its
inflation forecast. The private sector is endowed with a subjective causal
model, which is represented by a \textit{directed acyclic graph} (DAG) $%
G=(N,R)$, where $N$ is a set of nodes and $R$ is a set of directed links.
Each node in $N$ represents a variable. Some variables may be omitted from
the causal model, but we assume that $t$, $y$ and $\pi $ are always
represented. A link represents a perceived direct causal influence.

It will be convenient to enumerate the six variables $x_{1},...,x_{6}$, such
that $N$ is some subset of $\{1,...,6\}$. Denote $x=(x_{1},...,x_{6})$. For
every $M\subseteq \{1,...,6\}$, $x_{M}=(x_{i})_{i\in M}$ is the projection
of $x$ on the variables represented by $M$. Let $R(i)$ be the set of nodes $%
j $ such that $G$ includes the link $j\rightarrow i$. In other words, $R(i)$
represents the set of variables that are considered to be direct causes of
the variable $x_{i}$. I sometimes use boldface notation of a variable to
represent its label, such that for any variable $z$,its label is $\mathbf{z}$%
. For instance, when $R$ postulates that $y$ is a direct cause of $\pi $, we
can rewrite this as $\mathbf{y}\in R(\mathbf{\pi })$.

To arrive at an inflation forecast, the private sector \textquotedblleft
estimates\textquotedblright\ its model by performing measurements that
quantify the postulated causal links, and combining the results of these
measurements in accordance with the causal model. In other words, the
private sector \textit{fits} the causal model $G$ to the joint distribution $%
p$. Formally, the private sector's belief is defined as follows:%
\begin{equation}
p_{G}(x_{N})=\dprod\limits_{i\in N}p(x_{i}\mid x_{R(i)})
\label{factorization}
\end{equation}%
This is a \textit{Bayesian-network factorization formula}: it factorizes the
objective joint distribution $p$ (defined over $x_{N}$) into a product of
conditional-probability terms. Each term is extracted from $p$, but the
terms are put together in a way that may result in $p_{G}\neq p$. As long as 
$p$ has full support, $p_{G}$ is a well-defined distribution, which induces
well-defined marginal and conditional distributions. We say that a
distribution $p$ is consistent with $G$ if $p_{G}=p$. Indeed, the objective
joint distribution $p$ is consistent with the following DAG, denoted $%
G^{\ast }$:%
\[
\begin{array}{ccccc}
\theta & \rightarrow & a & \rightarrow & \pi \\ 
\downarrow & \swarrow &  &  & \downarrow \\ 
t & \rightarrow & e & \rightarrow & y%
\end{array}%
\]

One interpretation of $p_{G}$ is that it is the limit belief of a
discrete-time process of Bayesian updating, where the prior has full support
over all distributions that are consistent with $G$, and the private sector
observes an independent draw from $p$ at every period (see Spiegler
(2020b)). Another interpretation, which is appropriate when $p$ is
multivariate normal, is that $G$ represents a recursive system of linear
regression equations, where every term $p(x_{i}\mid x_{R(i)})$ in (\ref%
{factorization}) corresponds to a linear regression equation for $x_{i}$, in
which the set of regressors is $R(i)$. In this case, (\ref{factorization})
defines a so-called \textit{Gaussian Bayesian network} (Koller and Friedman
(2009, Ch. 7)). Each equation is estimated via OLS against an arbitrarily
large sample, and the estimated equations are put together in accordance
with the assumption that the noise terms in the regression equations are
independent. This assumption is wrong when $G$ is a misspecified causal
model.

Because the private sector perceives statistical regularities through the
prism of an incorrect model, the subjective belief $p_{G}$ may
systematically distort the correlation structure of the actual distribution $%
p$.

The private sector uses its belief $p_{G}$ to obtain the inflation forecast:%
\begin{equation}
e=E_{G}(\pi \mid t)=\dint\nolimits_{\pi }\pi dp_{G}(\pi \mid t)
\label{private sector forecast}
\end{equation}%
For instance, if $G$ is%
\[
\begin{array}{ccccc}
\theta & \rightarrow & a & \rightarrow & y \\ 
& \searrow & \downarrow &  & \downarrow \\ 
&  & t &  & \pi%
\end{array}%
\]%
then 
\[
E_{G}(\pi \mid t)=\int_{\theta ,a,y,\pi }p(\theta ,a\mid t)p(y\mid a)p(\pi
\mid a)\pi 
\]%
As we shall see, because $G$ is a misspecified model, $e$ may be sensitive
to all features of the joint distribution of $p$, including those that would
not matter under rational expectations.\medskip

\begin{definition}[Equilibrium]
The pair $(a(\theta ))_{\theta }$ and $(e(t))_{t}$ constitutes an $%
equilibrium$ if: $(i)$ for every $\theta $, $a(\theta )$ maximizes the
central bank's expected payoff given $(e(t))_{t}$; and $(ii)$ $(e(t))_{t}$
is given by (\ref{private sector forecast}).\medskip
\end{definition}

\begin{definition}[Linear equilibrium]
An equilibrium $(a(\theta ))_{\theta }$ and $(e(t))_{t}$ is $linear$ if $%
a(\theta )$ and $e(t)$ are linear functions.\medskip
\end{definition}

Throughout the paper, I restrict attention to linear equilibria. This
ensures that the joint distribution $p$ is multivariate normal.\medskip

\noindent \textit{Rational-expectations benchmark}

\noindent As a benchmark, suppose that the private sector has rational
expectations. That is, its subjective DAG is $G^{\ast }$ itself.
Furthermore, suppose it is fully informed of $\theta ,a$. Then, its
inflation forecast conditional on $\theta ,a$ is%
\[
e=E(\pi \mid \theta ,a)=E(\pi \mid a)=a 
\]%
Therefore,%
\[
y=a-\lambda a+\varepsilon +\eta 
\]%
The central bank's expected cost from an action $a$ given the realization of 
$\theta $ is:%
\[
(1-\lambda )^{2}a^{2}+(a-\theta )^{2}+2\sigma _{\varepsilon }^{2}+\sigma
_{\eta }^{2} 
\]%
Therefore, the central bank's optimal action given $\theta $ is%
\[
a^{\ast }(\theta )=\frac{1}{(1-\lambda )^{2}+1}\theta 
\]

The solution $a^{\ast }(\theta )$ is not fully responsive to changes in $%
\theta $. However, as $\lambda $ tends to one - such that only unanticipated
inflation matters for output - the central bank's policy approaches fully
flexible targeting. Another key observation is that $a^{\ast }(\theta )$ is
invariant to $\sigma _{\varepsilon }^{2}$ and $\sigma _{\eta }^{2}$.\medskip
\medskip

\noindent \textit{Discussion: Model-based forecasts}

\noindent The private sector's forecasting method consists of two steps: if
first estimates a causal model and then uses the estimated model to generate
an inflation forecast. The exact interpretation of this process depends on
whether we think of private-sector agents as lay persons who reason about
macroeconomic processes informally and intuitively, or as professional
forecasters. Under the former interpretation, the idea is that agents have
some qualitative causal perceptions of how macroeconomic variables relate to
one another, and they interpret observational data in light of these
perceptions.

As to the latter interpretation, it is clear that not all professional
forecasting is model-based. However, models do play a role in forecasting.
One virtue of models is that they ground the forecast in a \textquotedblleft
story\textquotedblright\ that can be communicated to other parties (e.g. see
Edge et al. (2008)). Another advantage of a model is that it is a simple
vehicle for making a multitude of conditional predictions. See, for example,
the following quote from a recent speech by Stanley Fisher:
\textquotedblleft The economy is an extremely complicated mechanism, and
every macroeconomic model is a vast simplification of reality\ldots the
large scale of FRB/US is an advantage in that it can perform a wide variety
of computational \textquotedblleft what if\textquotedblright\
experiments.\textquotedblright \footnote{%
See https://www.federalreserve.gov/newsevents/speech/fischer20170211a.htm.}
For a critical discussion of theory-based forecasting, see Giacomini (2015).

Thus, we should not think of the private sector as being exclusively
interested in predicting inflation - otherwise, there would be no need to
rely on a model; the private sector could simply measure the steady-state
conditional expectation $E(\pi \mid t)$. Instead, the private sector employs
a model as a multi-purpose vehicle for making sense of statistical
regularities in the macroeconomic environment, and for making conditional
predictions as the occasion arises. The model is a knowing simplification of
a complex environment, and therefore the private sector expects it to get
some things wrong, and would be unfazed if some of the model's assumptions
and predictions turned out to be incorrect (defending this stance with
familiar statements like \textquotedblleft every model is
wrong\textquotedblright\ or \textquotedblleft it takes a model to beat a
model\textquotedblright ).

\section{The Basic Result}

The model assumes that the private sector has a correct causal model of the
joint behavior of the variables $\theta ,a,t,e$. The private sector's error
is that it gets the direct causes of $\pi $ or $y$ wrong. In the true model $%
G^{\ast }=(N^{\ast },R^{\ast })$, $\mathbf{y}\notin R^{\ast }(\mathbf{\pi })$
- that is, output is not a direct cause of inflation, because causality runs
in the opposite direction. Our first result establishes that if the private
sector's causal model $G$ shares this property, its inflation forecast in a
linear equilibrium is unaffected by $\sigma _{\varepsilon }^{2}$ and $\sigma
_{\eta }^{2}$. Moreover, the equilibrium can be \textquotedblleft
rationalized\textquotedblright : it can be obtained by assuming that $%
G=G^{\ast }$ while modifying the private sector's signal function, given by $%
(p(t\mid \theta ,a))$.\medskip

\begin{proposition}
\label{prop_no_reverse}Suppose $G=(N,R)$ satisfies $\mathbf{y}\notin R(%
\mathbf{\pi })$. Then:\newline
\noindent $(i)$ Any linear equilibrium is invariant to $(\sigma
_{\varepsilon }^{2},\sigma _{\eta }^{2})$.\newline
\noindent $(ii)$ Any linear equilibrium can be obtained from a specification
of the primitives in which the private sector's causal model is $G^{\ast }$,
and the private sector is either fully informed or entirely uninformed of $%
(\theta ,a)$.
\end{proposition}

\begin{proof}
$(i)$ By assumption, $G$ coincides with $G^{\ast }$ over the nodes $\theta
,a,t,e$. Therefore,%
\begin{eqnarray*}
E_{G}(\pi &\mid &t)=\sum_{\theta ,a,e}p_{G}(\theta ,a,e\mid t)\sum_{\pi
}p_{G}(\pi \mid \theta ,a,t,e)\pi \\
&=&\sum_{\theta ,a,e}p_{G}(\theta ,a,e\mid t)\sum_{\pi }p_{G}(\pi \mid x_{R(%
\mathbf{\pi })})\pi \\
&=&\sum_{\theta ,a,e}p(\theta ,a,e\mid t)\sum_{\pi }p_{G}(\pi \mid x_{R(%
\mathbf{\pi })})\pi
\end{eqnarray*}%
By assumption, $R(\mathbf{\pi })\subseteq \{\mathbf{\theta },\mathbf{a},%
\mathbf{t},\mathbf{e}\}$. Therefore,%
\[
p_{G}(\pi \mid x_{R(\mathbf{\pi })})=\sum_{a^{\prime }}p(a^{\prime }\mid
x_{R(\mathbf{\pi })})p(\pi \mid a^{\prime }) 
\]%
such that%
\begin{equation}
E_{G}(\pi \mid t)=\sum_{x_{R(\mathbf{\pi })}}p(x_{R(\mathbf{\pi })}\mid
t)\sum_{a^{\prime }}p(a^{\prime }\mid x_{R(\mathbf{\pi })})\sum_{\pi }p(\pi
\mid a^{\prime })\pi  \label{EG(pi|t)}
\end{equation}%
Since 
\[
\sum_{\pi }p(\pi \mid a^{\prime })\pi =E(a^{\prime }+\varepsilon )=a^{\prime
} 
\]%
we obtain%
\[
E_{G}(\pi \mid t)=\sum_{x_{R(\mathbf{\pi })}}p(x_{R(\mathbf{\pi })}\mid
t)\sum_{a^{\prime }}p(a^{\prime }\mid x_{R(\mathbf{\pi })})a^{\prime } 
\]%
The terms $p(x_{R(\mathbf{\pi })}\mid t)$ and $p(a^{\prime }\mid x_{R(%
\mathbf{\pi })})$ only derive from $(p(\theta ,a,t,e))$, and are therefore
invariant to $(\sigma _{\varepsilon }^{2},\sigma _{\eta }^{2})$.\medskip

\noindent $(ii)$ Once again, recall that by assumption, $R(\mathbf{\pi }%
)\subseteq \{\mathbf{\theta },\mathbf{a},\mathbf{t},\mathbf{e}\}$. In a
linear equilibrium, $a$ is a deterministic function of $\theta $, and $e$ is
a deterministic function of $t$. Therefore, if $R(\mathbf{\pi })$ includes $%
\mathbf{a}$ ($\mathbf{\theta }$), it is immaterial whether it also includes $%
\mathbf{\theta }$ ($\mathbf{a}$). Likewise, if $R(\mathbf{\pi })$ includes $%
\mathbf{t}$ ($\mathbf{e}$), it is immaterial whether it also includes $%
\mathbf{e}$ ($\mathbf{t}$).

It follows that we only need to examine three cases. First, suppose that $%
\mathbf{a\in }R(\mathbf{\pi })$ (or, equivalently, $\mathbf{\theta \in }R(%
\mathbf{\pi })$). Then,%
\begin{eqnarray*}
E_{G}(\pi &\mid &t)=\sum_{\theta ,a}p(\theta ,a,e\mid t)E(\pi \mid x_{R(%
\mathbf{\pi })}) \\
&=&\sum_{\theta ,a}p(\theta ,a,e\mid t)E(\pi \mid a) \\
&=&\sum_{a}p(a\mid t)E(\pi \mid a) \\
&=&E(\pi \mid t)
\end{eqnarray*}
where the second equality follows from the fact that under the objective
distribution, $\pi \perp (\theta ,t,e)\mid a$. It follows that if $\mathbf{%
a\in }R(\mathbf{\pi })$, the private sector's inflation forecast is
consistent with rational expectations, even if we do not change its signal
function.

Second, suppose that $\{\mathbf{a},\mathbf{\theta }\}\cap R(\mathbf{\pi })$
is empty, but $\mathbf{t\in }R(\mathbf{\pi })$ (or, equivalently, $\mathbf{%
e\in }R(\mathbf{\pi })$). Then, by definition,%
\[
E_{G}(\pi \mid t)=E(\pi \mid t) 
\]%
Then, as in the first case, the private sector's inflation forecast is
consistent with rational expectations, even if we do not change its signal
function.

The only remaining case is that $R(\mathbf{\pi })$ is empty. In this case,%
\[
E_{G}(\pi \mid t)=E(\pi ) 
\]%
This inflation forecast can be generated in a rational-expectations model,
in which the private sector is entirely uninformed of $\theta ,a$ - i.e., we
modify $p(t\mid \theta ,a)$ into $p^{\prime }(t\mid \theta ,a)$ that is
constant in $\theta ,a$.\medskip
\end{proof}

This result establishes that nothing \textquotedblleft
interesting\textquotedblright\ happens in this model when the private
sector's model does not invert the causal link between inflation and output.
In fact, the only departure from the rational-expectations benchmark occurs
when the private sector's model postulates that none of the other relevant
variables are a direct cause of inflation. In this case, the private sector
acts as if it is entirely uninformed of $\theta $ and $a$. It is easy to
show that in the unique linear equilibrium, $e=0$ with certainty, while the
central bank's policy is 
\[
a(\theta )=\frac{\theta }{2}
\]%
That is, the private sector's expectations are entirely rigid, as in the
rational-expectations benchmark with $\lambda =0$. This results in an overly
rigid central-bank policy.

\section{Examples}

In this section I analyze in detail two examples in which the private
sector's causal model violates the condition in Proposition \ref%
{prop_no_reverse}. As a result, the noise terms in the equations for
inflation and output play a non-trivial role in the private sector's
inflation forecast, and therefore also in the central bank's equilibrium
policy.

In both examples, I assume that $t=a$ - i.e., the private sector observes $a$
before forming its inflation forecast. (The private sector does not observe $%
\theta $, but this is immaterial.) Since $t$ coincides with $a$, we can
remove it from the model, such that $G^{\ast }$ can be written as:%
\[
\begin{array}{ccccc}
\theta & \rightarrow & a & \rightarrow & \pi \\ 
&  & \downarrow &  & \downarrow \\ 
&  & e & \rightarrow & y%
\end{array}%
\]%
That is, the central bank's action $a$ is an immediate consequence of the
exogenous variable $\theta $; actual inflation $\pi $ and the private
sector's inflation forecast $e$ are conditionally independent consequences
of $a$; and these two variables are the direct causes of real output $y$.

We will now analyze two specifications of the private sector's DAG $G$.

\subsection{Case 1: $G:\protect\theta \rightarrow a\rightarrow \protect\pi %
\leftarrow y$}

This DAG tells a causal story, according to which inflation is a consequence
of two independent causes: real output and the central bank's policy. This
causal model postulates absolute \textit{monetary neutrality}: it admits no
path of causal links from $a$ to $y$. Thus, the private sector's causal
model distorts the true model\ $G^{\ast }$ in two ways: it omits the
expectations variable $e$ and inverts the causal link between inflation and
output.

The private sector's inflation forecast after observing the central bank's
action $a$ is%
\begin{equation}
E_{G}(\pi \mid a)=\dint\nolimits_{\pi }\pi p_{G}(\pi \mid
a)=\dint\nolimits_{\pi }\dint\nolimits_{y}\pi p(y)p(\pi \mid a,y)
\label{inflationforecast}
\end{equation}%
Compare this with the following way to write the rational-expectations
inflation forecast:%
\begin{equation}
E(\pi \mid a)=\dint\nolimits_{\pi }\dint\nolimits_{y}\pi p(y\mid a)p(\pi
\mid a,y)  \label{correct_forecast}
\end{equation}%
The discrepancy arises because $p_{G}(\pi \mid a)$ involves an implicit
expectation over $y$ \textit{without} conditioning on $a$. (Note that since $%
\pi ,y\perp \theta \mid a$ according to both $G^{\ast }$ and $G$, $\theta $
does not feature in either of the two expressions.) Because the term $p(y)$
in (\ref{inflationforecast}) is actually $not$ independent of $a$, a change
in the central bank's strategy can lead to a change in $E_{G}(\pi \mid a)$%
.\medskip

\begin{proposition}
There is a unique linear equilibrium, given by%
\begin{eqnarray}
e(a) &=&\frac{1-\beta }{1-\beta \lambda }a  \label{linear_eq_neutrality} \\
a(\theta ) &=&\frac{1}{\left( \frac{1-\lambda }{1-\beta \lambda }\right)
^{2}+1}\theta  \nonumber
\end{eqnarray}%
where%
\begin{equation}
\beta =\frac{\sigma _{\varepsilon }^{2}}{\sigma _{\varepsilon }^{2}+\sigma
_{\eta }^{2}}  \label{delta}
\end{equation}
\end{proposition}

\begin{proof}
Let us begin by deriving a formula for the private sector's inflation
forecast. Denote%
\[
e(a)=E_{G}(\pi \mid a)=\dint\nolimits_{y}p(y)E(\pi \mid a,y) 
\]%
Since $\pi =a+\varepsilon $, 
\[
E(\pi \mid a,y)=a+E(\varepsilon \mid a,y) 
\]%
By the Phillips Curve,%
\[
\varepsilon +\eta =y-a+\lambda e(a) 
\]%
For given $a$ and $y$, the R.H.S is a constant, whereas the L.H.S is a sum
of two independent variables that are normally distributed with mean zero.
Therefore, to calculate $E(\varepsilon \mid a,y)$, we can apply the standard
formula for $E(X\mid X+Z)$ when $X$ and $Z$ are independent normal
variables, and obtain%
\[
E(\varepsilon \mid a,y)=\beta (y-a+\lambda e(a)) 
\]%
where%
\[
\beta =\frac{\sigma _{\varepsilon }^{2}}{\sigma _{\varepsilon }^{2}+\sigma
_{\eta }^{2}} 
\]

We can now write%
\[
e(a)=\int_{y}p(y)\left[ a+\beta y-\beta a+\beta \lambda e(a)\right]
=(1-\beta )a+\beta \lambda e(a)+\beta E(y) 
\]%
Since $\pi =a+\varepsilon $ and $E(\varepsilon )=0$, $E(\pi )=E(a)$.
Plugging the Phillips curve, we obtain%
\[
E(y)=E(a)-\lambda E(e(a)) 
\]%
and therefore,%
\[
e(a)=(1-\beta )a+\beta \lambda e(a)+\beta \lbrack E(a)-\lambda E(e(a))] 
\]

This functional equation defines $e(a)$. Taking expectations, we obtain%
\begin{eqnarray*}
E(e(a)) &=&(1-\beta )E(a)+\beta \lambda E(e(a))+\beta E(a)-\beta \lambda
E(e(a)) \\
&=&E(a)=E(\pi )
\end{eqnarray*}%
Plug\ this into the expression for $e(a)$, and get:%
\[
e(a)=\frac{1-\beta }{1-\beta \lambda }a+\frac{\beta -\beta \lambda }{1-\beta
\lambda }E(a) 
\]%
Plugging the expression for $\beta $, we obtain%
\[
e(a)=\frac{\sigma _{\eta }^{2}}{\sigma _{\eta }^{2}+(1-\lambda )\sigma
_{\varepsilon }^{2}}a+\frac{(1-\lambda )\sigma _{\varepsilon }^{2}}{\sigma
_{\eta }^{2}+(1-\lambda )\sigma _{\varepsilon }^{2}}E(a) 
\]%
Plugging the expression we obtained for $e(a)$, the expected output
conditional on $a$ is%
\[
E(y\mid a)=\frac{1-\lambda }{1-\beta \lambda }a-\lambda \frac{\beta -\beta
\lambda }{1-\beta \lambda }E(a) 
\]

We can now write down the central bank's expected payoff from $a$ given $%
\theta $ as follows:%
\[
\left[ \frac{1-\lambda }{1-\beta \lambda }a-\lambda \frac{\beta -\beta
\lambda }{1-\beta \lambda }E(a)\right] ^{2}+(a-\theta )^{2}
\]%
Our definition of linear equilibrium adopts the interim approach - i.e.,
when finding the optimal $a$ given $\theta $, we need to take $E(a)$ \textit{%
as given}.\ (This does not matter in fact; we could also take the ex-ante
approach and the result would be the same, although this is an artifact of
the linearity of the model.) The first-order condition with respect to $a$ is

\[
\left( \frac{1-\lambda }{1-\beta \lambda }\right) \left[ \frac{1-\lambda }{%
1-\beta \lambda }a-\lambda \frac{\beta -\beta \lambda }{1-\beta \lambda }E(a)%
\right] +a-\theta =0
\]%
Taking the expectation with respect to $a$ and recalling that $E(\theta )=0$%
, it is immediate that $E(a)=0$. Therefore,%
\[
a(\theta )=\frac{1}{\left( \frac{1-\lambda }{1-\beta \lambda }\right) ^{2}+1}%
\theta 
\]%
We have now obtained the precise formulas for the unique linear equilibrium,
given by (\ref{linear_eq_neutrality}).\medskip 
\end{proof}

Thus, when $\lambda =1$, $e(a)=a=E(\pi \mid a)$ - i.e., the private sector's
inflation forecast coincides with the rational-expectations benchmark.
Despite the private sector's misspecified Phillips Curve, it ends up having
correct expectations and therefore the central bank's strategy coincides
with the rational-expectations benchmark.

However, when $\lambda <1$ - i.e., when anticipated inflation has real
effects - the inflation forecast is \textquotedblleft
rigid\textquotedblright\ in the sense of being a convex combination of the
correct conditional expected inflation $E(\pi \mid a)$ and the ex-ante,
unconditional expected inflation $E(\pi )$. As a result, the central bank's
equilibrium strategy is less responsive to $\theta $ than in the
rational-expectations benchmark.

To understand the reason for the private sector's partially responsive
expectations, recall that according to the DAG $G:\theta \rightarrow
a\rightarrow \pi \leftarrow y$, the private sector erroneously regards $y$
as an \textit{exogenous} variable that affects $\pi $, and therefore assigns
some weight to the ex-ante expected value of $y$ when forming its inflation
forecast. Because $y$ is in fact a consequence of $a$, the private sector
ends up assigning weight to the ex-ante expectation of $a$, thus failing to
fully condition on the actual realization of $a$.

The \textit{extent} of this failure depends on the relative magnitudes of $%
\sigma _{\varepsilon }^{2}$ and $\sigma _{\eta }^{2}$. As the Phillips
relation becomes less noisy (relative to the relation between the central
bank's action and inflation), the erroneous weight on $E(a)$ increases and
the deviation from rational expectations is exacerbated. This
\textquotedblleft expectational rigidity\textquotedblright\ leads to a
policy-attenuation effect that increases with $\sigma _{\varepsilon
}^{2}/\sigma _{\eta }^{2}$.

The intuition behind the comparative statics is as follows. The private
sector tries to account for fluctuations in $\pi $ (conditional on $a$) by
the variation in $y$, mistakenly treating the latter as exogenous. The
variables $\pi $ and $y$ are normally distributed. When $\sigma _{\eta
}^{2}/\sigma _{\varepsilon }^{2}$ is low, the definition of conditional
expectation of normal variables implies that $y$ gets a large weight in the
private sector's inflation forecast, and therefore the model
misspecification error generates a larger forecast error - which in turn
impels the central bank to adopt a more rigid policy.

\subsection{Case 2: $G:\protect\theta \rightarrow a\rightarrow y\rightarrow 
\protect\pi $}

This DAG tells a causal story, according to which the central bank's action
has a causal effect on real output, which in turn is the sole direct cause
of inflation. As in Case 1, this specification of $G$ inverts the true
causal link between output and inflation. However, unlike Case 1, this
causal model regards $y$ as an endogenous variable that mediates the effect
of monetary policy on inflation.

The private sector's inflation forecast after observing the central bank's
action $a$ is%
\begin{equation}
E_{G}(\pi \mid a)=\dint\nolimits_{\pi }\pi p_{G}(\pi \mid
a)=\dint\nolimits_{\pi }\dint\nolimits_{y}\pi p(y\mid a)p(\pi \mid y)
\label{case2_forecast}
\end{equation}%
Compare this expression with (\ref{correct_forecast}). The error in (\ref%
{case2_forecast}) lies in the term $p(\pi \mid y)$, which fails to condition 
$\pi $ on $a$. This failure reflects the assumption, embedded in $G$, that $%
\pi $ is independent of $a$ conditional on $y$. This assumption is
inconsistent with the true model $G^{\ast }$.

A linear equilibrium is \textit{strongly linear} if $e(a)$ is proportional
to $a$ and $a(\theta )$ is proportional to $\theta $. The next result
restricts attention to strongly linear equilibria, in order to avoid the
minor technical distraction of proving that every linear equilibrium must be
strongly linear (this could be demonstrated transparently in Case 1, but
unfortunately not in the present case).\medskip 

\begin{proposition}
There is a unique strongly linear equilibrium, given by%
\begin{eqnarray}
e(a) &=&\gamma a  \label{linear_eq_case2} \\
a(\theta ) &=&\delta \theta  \nonumber
\end{eqnarray}%
where $\gamma $ and $\delta $ are uniquely given by the pair of equations:%
\begin{eqnarray}
\gamma &=&\frac{(1-\gamma \lambda )^{2}\delta ^{2}\sigma _{\theta
}^{2}+(1-\gamma \lambda )\sigma _{\varepsilon }^{2}}{(1-\gamma \lambda
)^{2}\delta ^{2}\sigma _{\theta }^{2}+\sigma _{\varepsilon }^{2}+\sigma
_{\eta }^{2}}  \label{gamma} \\
\delta &=&\frac{1}{(1-\gamma \lambda )^{2}+1}  \nonumber
\end{eqnarray}
\end{proposition}

\begin{proof}
Let us begin by deriving a formula for the private sector's inflation
forecast. Denote%
\begin{equation}
e(a)=E_{G}(\pi \mid a)=\dint\nolimits_{y}p(y\mid a)E(\pi \mid y)
\label{e(a)_case2}
\end{equation}%
Guessing a\ strongly linear equilibrium, suppose $e(a)=\gamma a$ and $%
a(\theta )=\delta \theta $. Therefore,%
\[
y=a-\lambda \gamma a+\varepsilon +\eta 
\]%
Given $\theta $, the central bank effectively chooses $a$ to minimize%
\[
E(y^{2}\mid a)+(a-\theta )^{2}=(1-\gamma \lambda )^{2}a^{2}+(a-\theta )^{2}
\]%
The first-order condition with respect to $a$ pins down the expression for $%
\delta $ given by (\ref{gamma}).

Our remaining task is to derive the expression for $\gamma $. To do so, let
us first obtain an expression for $E(\pi \mid y)$. By the true equations for
inflation and output,%
\begin{eqnarray*}
E(\pi &\mid &y)=E(a+\varepsilon \mid a-\lambda e(a)+\varepsilon +\eta =y) \\
&=&E(a+\varepsilon \mid a-\lambda \gamma a+\varepsilon +\eta =y)
\end{eqnarray*}%
which is equal to%
\[
\frac{1}{1-\gamma \lambda }E[(1-\gamma \lambda )a\mid (1-\gamma \lambda
)a+\varepsilon +\eta =y]+E[\varepsilon \mid (1-\gamma \lambda )a+\varepsilon
+\eta =y] 
\]

The expression $(1-\gamma \lambda )a+\varepsilon +\eta $ is a sum of
independent, normally distributed variables with mean zero. Denote%
\[
\sigma _{a}^{2}=\delta ^{2}\sigma _{\theta }^{2} 
\]%
Then, using the conditional-expectation rule for such variables,%
\[
E[a\mid (1-\gamma \lambda )a+\varepsilon +\eta =y]=\frac{(1-\gamma \lambda
)\delta ^{2}\sigma _{\theta }^{2}}{(1-\gamma \lambda )^{2}\delta ^{2}\sigma
_{\theta }^{2}+\sigma _{\varepsilon }^{2}+\sigma _{\eta }^{2}}\cdot y 
\]%
and%
\[
E[\varepsilon \mid (1-\gamma \lambda )a+\varepsilon +\eta =y]=\frac{\sigma
_{\varepsilon }^{2}}{(1-\gamma \lambda )^{2}\delta ^{2}\sigma _{\theta
}^{2}+\sigma _{\varepsilon }^{2}+\sigma _{\eta }^{2}}\cdot y 
\]%
Therefore,%
\[
E(\pi \mid y)=\frac{(1-\gamma \lambda )\delta ^{2}\sigma _{\theta
}^{2}+\sigma _{\varepsilon }^{2}}{(1-\gamma \lambda )^{2}\delta ^{2}\sigma
_{\theta }^{2}+\sigma _{\varepsilon }^{2}+\sigma _{\eta }^{2}}\cdot y 
\]

Now, plug this expression in (\ref{e(a)_case2}) and obtain%
\[
e(a)=\frac{(1-\gamma \lambda )\delta ^{2}\sigma _{\theta }^{2}+\sigma
_{\varepsilon }^{2}}{(1-\gamma \lambda )^{2}\delta ^{2}\sigma _{\theta
}^{2}+\sigma _{\varepsilon }^{2}+\sigma _{\eta }^{2}}\cdot E(y\mid a) 
\]%
By the inflation and output equations,%
\begin{eqnarray*}
E(y &\mid &a)=E(\pi -\lambda e(a)+\eta \mid a) \\
&=&E(a-\lambda \gamma a+\varepsilon +\eta \mid a) \\
&=&(1-\gamma \lambda )a
\end{eqnarray*}%
Therefore, 
\[
e(a)=\frac{(1-\gamma \lambda )\delta ^{2}\sigma _{\theta }^{2}+\sigma
_{\varepsilon }^{2}}{(1-\gamma \lambda )^{2}\delta ^{2}\sigma _{\theta
}^{2}+\sigma _{\varepsilon }^{2}+\sigma _{\eta }^{2}}\cdot (1-\gamma \lambda
)a=\gamma a 
\]%
where $\gamma $ is given by (\ref{gamma}).

It remains to establish uniqueness of $(\gamma ,\delta )$. Clearly, $\delta $
is a monotone function of $\gamma $ in the relevant range. Therefore, it
suffices to show that $\gamma $ is unique. Plugging the equation for $\delta 
$ in the equation for $\gamma $, we obtain%
\[
\gamma =\frac{\frac{(1-\gamma \lambda )^{2}}{[(1-\gamma \lambda )^{2}+1]^{2}}%
\sigma _{\theta }^{2}+(1-\gamma \lambda )\sigma _{\varepsilon }^{2}}{\frac{%
(1-\gamma \lambda )^{2}}{[(1-\gamma \lambda )^{2}+1]^{2}}\sigma _{\theta
}^{2}+\sigma _{\varepsilon }^{2}+\sigma _{\eta }^{2}}
\]

It can be checked that the R.H.S of this equation is continuously and
monotonically decreasing in $\gamma $ in the range $[0,1]$. Moreover, it is
strictly positive when $\gamma =0$ and strictly below $1$ when $\gamma =1$.
Therefore, there is a unique intersection between the curve given by the
R.H.S and the line given by the L.H.S.\medskip 
\end{proof}

As in Case 1, the variance of the noise terms in the inflation and output
equations plays a role in the strongly linear equilibrium in the present
case. However, this role is different in a number of ways.

First, the effect of the noisiness of the Phillips Curve (i.e. $\sigma
_{\eta }^{2}$) on expectational rigidity goes in the opposite direction.
Expectations become more rigid as the Phillips relation becomes \textit{less}
precise. The intuition is that in Case 2, $G$ posits that the effect of $a$
on $\pi $ is mediated by $y$. As the statistical relation between $\pi $ and 
$y$ becomes noisier, the perceived mediation effect is attenuated, and this
means that the perceived indirect effect of $a$ on $\pi $ becomes weaker.

Second, unlike Case 1, the private sector's expectations depart from the
rational-expectations benchmark even when $\lambda =1$. That is, even when
anticipated inflation fully offsets the effect of inflation on output, the
private sector's causal model leads to incorrect, rigid inflation forecasts.
In particular, when $\lambda =1$ and $\sigma _{\varepsilon }^{2}$ (which
measures the noisiness of the mapping from $a$ to $\pi $) gets large, $%
\gamma $ approaches $\frac{1}{2}$ - i.e., the inflation forecast is only
\textquotedblleft half-responsive\textquotedblright\ to changes in $a$.

Finally, unlike Case 1, the variance of the exogenous inflation target $%
\theta $ plays a role in the characterization of the linear equilibrium. As $%
\sigma _{\theta }^{2}$ becomes larger, the private sector's expectations
become less rigid, and as a result the central bank's policy approaches the
rational-expectations benchmark. The intuition is that larger fluctuations
in $\theta $ translate to larger fluctuations in $a$. This in turn implies
that the private sector's interpretation of fluctuations in inflation and
output becomes dominated by the fluctuations in $a$, which is the correct
interpretation; the misperception of the role of $\sigma _{\eta }^{2}$ and $%
\sigma _{\varepsilon }^{2}$ becomes less important for the private sector's
expectations.

\section{Discussion}

The comparison between Cases 1 and 2 shows that the details of the private
sector's causal misinterpretation of the Phillips Curve matter for the
qualitative properties of the equilibrium in the interaction between the
central bank and the private sector. However, the broad lesson that $\sigma
_{\eta }^{2}$ and $\sigma _{\varepsilon }^{2}$ matter for the equilibrium
holds in both cases because they share the reverse-causality error that
misperceives the direction of causality between output and inflation.

Note, however, that not every $G$ that exhibits this reverse causality will
give rise to this effect. For example, suppose $R(\mathbf{\pi })=\{\mathbf{y}%
\}$ and $R(\mathbf{y})=\emptyset $. Then, $G$ admits no causal path from $a$
or $e$ to $\pi $. As in earlier examples that exhibited this feature, the
private sector's equilibrium expectations will be entirely rigid: $e=0$ with
probability one, leading to the most rigid central-bank policy that is
possible in this model, without any sensitivity to $\sigma _{\eta }^{2}$ and 
$\sigma _{\varepsilon }^{2}$. Another example is $R(\mathbf{\pi })=\{\mathbf{%
a},\mathbf{y}\}$ and $R(\mathbf{y})=\{\mathbf{a}\}$. In this case, $G$
induces rational expectations. This can be shown algebraically, or using
basic equivalence results from the Bayesian-network literature (see Spiegler
(2016,2020b) for details). Thus, reverse causality is a necessary condition
for the noise-sensitivity of central-bank policy, but not a sufficient one.

The noise-sensitivity of linear equilibria is of interest in light of the
\textquotedblleft \textit{caution principle}\textquotedblright\ attributed
to Brainard (1967), which has been suggested as an explanation of
monetary-policy attenuation. A high $\sigma _{\varepsilon }^{2}$ captures a
situation in which the central bank is highly uncertain about the effect of
monetary policy on inflation. However, given our conventional
linear-quadratic parameterization of the model, this uncertainty has no
effect on the central bank's policy under rational expectations. However,
when the private sector's model reverses the causal link between inflation
and output, high $\sigma _{\varepsilon }^{2}$ can lead to an attenuated
policy - although, as we saw, the extent of this effect depends on details
of the private sector's model. The reason is that $\sigma _{\varepsilon }^{2}
$ affects the private sector's interpretation of inflation fluctuations, and
therefore its inflation forecast conditional on its information. The lesson
is that Phillips-Curve disagreements between the central bank and the
private sector can rationalize policy distortions due to uncertainty about
the consequences of monetary policy - even in settings that otherwise could
not give rise to policy distortions.

The policy distortion that emerged in this paper takes the form of \textit{%
attenuation} - i.e., the policy is less responsive to exogenous shocks than
in the rational-expectations benchmark. However, I do not want to stress
this point too strongly. The important effect that our analysis has
highlighted is the attenuated response of the private sector's inflation
forecast, as a result of its causally misperceived Phillips Curve. The fact
that this effect translated to a policy-attenuation effect is an artifact of
how the model treated the central bank's private information $\theta $ -
namely, as an exogenous inflation target. Under alternative specifications
of the central bank's preferences, rigid private-sector forecasts would lead
to \textit{over-responsive} central-bank policy. One could also imagine a
model in which $\theta $ has no direct payoff relevance, and instead plays a
role as a parameter in the inflation/output equations.\footnote{%
I thank Stephane Dupraz for making these observations.} Depending on how the
private sector incorporates this dependency into its causal model, such a
specification could lead to a more complicated, less transparent relation
between the private sector's forecast errors and the central bank's policy
distortion. Since paper is a first attempt to incorporate causal
disagreements into familiar models of monetary policy, I prioritized
pedagogical considerations and refrained from adding these complications.
However, I believe they merit further study.

\end{document}